\documentclass{amsart}
\usepackage{
    mathtools,
    hyperref,
    booktabs,
}

\newtheorem*{thm}  {Theorem}
\newtheorem*{lem}  {Lemma}
\newtheorem*{cor}  {Corollary}
\theoremstyle{definition}
\newtheorem*{dfn}  {Definition}
\newtheorem*{rmk}  {Remark}
\newtheorem*{prob} {Problem}

\DeclarePairedDelimiterX\set[1]\lbrace\rbrace{\,\def\given{\delimsize:}#1\,} % usage: \set{ 2n \given n > 0 }, \set[\big]{ \frac{p}{q} \given p, q > 0 }
\DeclareMathOperator{\wt}{wt}
\newcommand{\ZZ}{\mathbb{Z}}
\newcommand{\RR}{\mathbb{R}}
\renewcommand{\vec}[1]{#1}
\newcommand{\code}[1]{#1}
\newcommand{\e}[1]{e\left(\frac{j#1}{m}\right)}
\newcommand{\me}[1]{e\left(-\frac{j#1}{m}\right)}

\begin{document}

% Front matter
\title[Weight enumerator of linear-congruence codes]
      {An explicit formula for a\\ weight enumerator of linear-congruence codes}
\author[T.~Sakurai]{\href{https://orcid.org/0000-0003-0608-1852}{Taro Sakurai}}
\address{Department of Mathematics and Informatics, Graduate School of Science, Chiba University, 1-33, Yayoi-cho, Inage-ku, Chiba-shi, Chiba, 263-8522 Japan}
\email{tsakurai@math.s.chiba-u.ac.jp}

\subjclass[2010]{%
    \href{https://zbmath.org/classification/?q=94B60}{94B60}
    (\href{https://zbmath.org/classification/?q=05A15}{05A15},
    \href{https://zbmath.org/classification/?q=11L15}{11L15})}
\date{\today}
\begin{abstract}
    \noindent % The first line of the first paragraph is not indented, but first lines of subsequent paragraphs are.
    An explicit formula for a weight enumerator of linear-congruence codes is provided.
    This extends the work of Bibak and Milenkovic [IEEE ISIT (2018) 431--435] addressing the binary case to the non-binary case.
    Furthermore, the extension simplifies their proof and provides a complete solution to a problem posed by them.
\end{abstract}

\maketitle

% Body
\section*{Introduction}
\noindent % the very first line of the paper is not indented
Throughout this article, \( n \) and \( m \) denote positive integers,
\( b \) denotes an integer
and \( \ZZ_q \coloneqq \{0, 1, \dotsc, q-1\} \subset \ZZ \) for a positive integer \( q \).
We will use \( n \) for a code length, \( m \) for a modulus,
\( b \) for a defining parameter of a code
and \( \ZZ_q \) for a code alphabet.

\begin{dfn}
    Let \( \vec{a} = (a_1, \dotsc, a_n) \in \ZZ^n \) and \( b \in \ZZ \).
    The set \( \code{C} \) of all the solutions
    \( \vec{x} = (x_1, \dotsc, x_n) \in \ZZ_q^n \) for a linear congruence equation
    \begin{equation}
        \label{eq: ax = b}
        \vec{a}\cdot\vec{x} \equiv b \pmod m
    \end{equation}
    is said to be a \emph{linear-congruence code} where \( \vec{a}\cdot\vec{x} \coloneqq a_1x_1 + \dotsb + a_nx_n \).
    A linear-congruence code \( \code{C} \) is called \emph{binary} when \( q = 2 \).
\end{dfn}

Several deletion-correcting codes which have been studied are linear-congruence codes;
    the Varshamov-Tenengol'ts codes~\cite{VT65}, the Levenshtein codes~\cite{Lev66}, the Helberg codes~\cite{HF02}, the Le-Nguyen codes~\cite{LN16}, the construction \( C' \) of Hagiwara~\cite{Hag17} (for some parameters), the consecutively systematic encodable codes and the ternary integer codes in~\cite[Examples~II.1 and II.5]{Hag16} fall into this category (Table).
\begin{table}[pbth]
    \renewcommand\thetable{\!\!} % to suppress only the number but retaining the name ("Table")
    \centering
    \caption{Examples of linear-congruence codes}
        \begin{tabular}{p{4.4cm}rrrp{2.5cm}}
        \toprule
        Linear-congruence code\footnotemark[1] & \( q \) & \( (a_1, \dotsc, a_n) \) & \( m \)         & Constraints                                            \\
        \midrule
        Varshamov-Tenengol'ts code             & \( 2 \) & \( (  1, \dotsc,   n) \) & \( n+1 \)       &                                                        \\
        Levenshtein code                       & \( 2 \) & \( (  1, \dotsc,   n) \) & \( m \)         & \( m \geq n + 1 \)                                     \\
        Helberg code\footnotemark[2]           & \( 2 \) & \( (v_1, \dotsc, v_n) \) & \( v_{n+1} \)   & \( s \in \ZZ_{>0} \)                                   \\
        Le-Nguyen code\footnotemark[3]         & \( q \) & \( (w_1, \dotsc, w_n) \) & \( m \)         & \( m \geq w_{n+1} \), \newline \( s \in \ZZ_{>0} \)    \\
        Construction \( C' \)\footnotemark[4]  & \( 2 \) & \( (c_1, \dotsc, c_n) \) & \( n \)         & \( b \not\equiv 0, n(n+1)/2 \pmod n \)                 \\
        Consecutively systematic \newline
            encodable codes\footnotemark[5]    & \( 2 \) & \( (b_1, \dotsc, b_n) \) & \( 2^{s+1} \)   & \( b = 0 \), \( s \in \ZZ_{>0} \), \newline \( 0 < n - s < 2^{s-1} \)     \\
        Ternary integer code\footnotemark[6]   & \( 3 \) & \( (t_1, \dotsc, t_n) \) & \( 2^{n+1}-1 \) &                                                        \\
        \bottomrule
    \end{tabular}

\end{table}
\footnotetext[1]{The defining parameter \( b \) for the codes in the table takes an arbitrary value unless otherwise stated.}
\footnotetext[2]{The sequence \( (v_i) = (v_i(s)) \) is defined by \( v_i = 0 \) (\( i \leq 0 \)) and \( v_i = 1 + \sum_{j=1}^s v_{i-j} \) (\( i \geq 1 \)).}
\footnotetext[3]{The sequence \( (w_i) = (w_i(q, s)) \) is defined by \( w_i = 0 \) (\( i \leq 0 \)) and \( w_i = 1 + (q - 1)\sum_{j=1}^s w_{i-j} \) (\(i \geq 1 \)).}
\footnotetext[4]{The sequence \( (c_i) = (c_i(n)) \) is defined by \( c_{2i-1} = i \) (\(1 \leq i \leq \lfloor \frac{n+1}{2} \rfloor \)) and \( c_{2i} = n - i + 1 \) (\(1 \leq i \leq \lfloor \frac{n}{2} \rfloor \)).}
\footnotetext[5]{The sequence \( (b_i) = (b_i(s)) \) is defined by \( b_i = 2^{i-1} \) (\(1 \leq i \leq s \)) and \( b_i = 2^{s-1} + i - s \) (\( i > s \)).}
\footnotetext[6]{The sequence \( (t_i) \) is defined by \( t_i = 2^i - 1 \) (\( i \geq 1 \)).}

The following problem concerning the size of a linear-congruence code---the number of solutions for a linear congruence equation \eqref{eq: ax = b}---is posed by Bibak and Milenkovic.
\begin{prob}[Bibak-Milenkovic~\cite{BM18}]
    Give an explicit formula for the size of a linear-congruence code.
\end{prob}
Finding an explicit formula would be a first step toward understanding the asymptotic behavior of the size of a linear-congruence code.
Bibak and Milenkovic provide a solution to the problem for the binary case.
In this article, we provide a complete solution to the problem with a simple proof, which improves the argument of Bibak and Milenkovic.
Actually, what we will show is how the Hamming weights of the solutions for a linear congruence equation distribute.
This immediately gives an expression of the size of a linear-congruence code involving exponential sums---Weyl sums of degree one.

To state the main theorem we need notation which will be standard.
\begin{dfn}
    For a code \( \code{C} \subseteq \ZZ_q^n \), we define a polynomial \( W_\code{C}(z) \) by
    \begin{equation*}
        W_\code{C}(z)
        \coloneqq \sum_{\vec{x} \in \code{C}} z^{\wt(\vec{x})}
        = \sum_{i=0}^n A_i(\code{C}) z^i,
    \end{equation*}
    where \( \wt(\vec{x}) \) denotes the Hamming weight and
    \begin{equation*}
        A_i(\code{C}) \coloneqq
        \lvert\set{ \vec{x} \in \code{C} \given \wt(\vec{x}) = i }\rvert \qquad (0 \leq i \leq n).
    \end{equation*}
    The polynomial \( W_\code{C}(z) \) is said to be the (non-homogeneous) \emph{weight enumerator} of the code \( \code{C} \).
\end{dfn}
Following custom due to Vinogradov in additive number theory, % "Weyl Sum" MathWorld
\( e(\alpha) \) denotes \( e^{2\pi\alpha\sqrt{-1}} \) for \(\alpha \in \RR \).
Now we are in position to state our main theorem.
\begin{thm}
    Let \( \vec{a} = (a_1, \dotsc, a_n) \in \ZZ^n \) and \( b \in \ZZ \).
    Then the weight enumerator \( W_\code{C}(z) \) of the linear-congruence code
    \begin{equation}
        \label{eq: LCC}
        \code{C} \coloneqq \set{ \vec{x} \in \ZZ_q^n \given \vec{a}\cdot\vec{x} \equiv b \pmod m }
    \end{equation}
    is given by
    \begin{equation}
        W_\code{C}(z) =
        \frac{1}{m}\sum_{j=1}^m \me{b}
            \prod_{i=1}^n\left(1 + z\e{a_i} + \dotsb + z\e{a_i(q-1)}\right).
    \end{equation}
\end{thm}

\begin{cor}
    With the same notation as above, the size of the code \( \code{C} \)
    is given by
    \begin{equation*}
        \lvert\code{C}\rvert =
        \frac{1}{m}\sum_{j=1}^m \me{b}
            \prod_{i=1}^n\left(1 + \e{a_i} + \dotsb + \e{a_i(q-1)}\right).
    \end{equation*}
\end{cor}

\section*{Proof of Theorem}
% Should I mention Hardy-Littlewood circle method?
The only lemma we need to prove the main theorem is the following trivial one.
\begin{lem}
    \begin{equation*}
        \frac{1}{m}\sum_{j=1}^m \e{b}
        = \begin{cases}
            1 & \textup{if \( b     \equiv 0 \pmod m \)} \\
            0 & \textup{if \( b \not\equiv 0 \pmod m \)}.
        \end{cases}
    \end{equation*}
\end{lem}

\begin{proof}[Proof of Theorem]
    The proof is straightforward:
    \begin{align*}
        &\frac{1}{m}\sum_{j=1}^m \me{b}
            \prod_{i=1}^n\left(1 + z\e{a_i} + \dotsb + z\e{a_i(q-1)}\right) \\
        &\qquad=
        \frac{1}{m}\sum_{j=1}^m \me{b}
            \prod_{i=1}^n \sum_{x_i \in \ZZ_q} z^{\wt(x_i)}\e{a_ix_i} \\
        &\qquad=
        \frac{1}{m}\sum_{j=1}^m \me{b}
            \sum_{(x_1, \dotsc, x_n) \in \ZZ_q^n} \prod_{i=1}^n z^{\wt(x_i)}\e{a_i x_i} \\
        &\qquad=
        \frac{1}{m}\sum_{j=1}^m \me{b}
            \sum_{\vec{x} \in \ZZ_q^n} z^{\wt(\vec{x})}\e{\vec{a}\cdot\vec{x}} \\
        &\qquad=
        \sum_{\vec{x} \in \ZZ_q^n}
            \left(\frac{1}{m}\sum_{j=1}^m
                \e{(\vec{a}\cdot\vec{x} - b)} \right)
                     z^{\wt(\vec{x})} \\
        &\qquad=
        \sum_{\vec{x} \in \code{C}}z^{\wt(\vec{x})} \qquad (\text{By Lemma.}) \\
        &\qquad= W_\code{C}(z).
        \qedhere
    \end{align*}
\end{proof}

\begin{rmk}
    The original proof by Bibak and Milenkovic~\cite{BM18} for the binary case uses a theorem of Lehmer~\cite{Leh13}, which states a linear congruence equation
    \begin{equation*}
        \vec{a}\cdot\vec{x} \equiv b \pmod m
    \end{equation*}
    defined by \( \vec{a} = (a_1, \dotsc, a_n) \in \ZZ^n \) and \( b \in \ZZ \) has a solution \( \vec{x} \in \ZZ_m^n \) if and only if \( \gcd(a_1, \dotsc, a_n, m) \) divides \( b \).
    Consequently, their result is stated depending on whether \( \gcd(a_1, \dotsc, a_n, m) \) divides \( b \) or not.
    By contrast, our result does not refer to \( \gcd(a_1, \dotsc, a_n, m) \) because our proof does not rely on the Lehmer theorem.
\end{rmk}

\section*{Acknowledgments}
The author thanks Professor Manabu Hagiwara for drawing the author's attention to the work of Bibak and Milenkovic and his invaluable help during the preparation of this article.
This work is partially supported by KAKENHI(B) 18H01435, 16K12391 and 16K06336.

\end{document}